\RequirePackage{xcolor}
\documentclass[letterpaper, 10 pt,  journals]{ieeecolor}  % Comment this line out
\definecolor{subsectioncolor}{rgb}{0,0.541,0.855}
\usepackage{float}

\usepackage{TLC}
\usepackage{xurl}
\usepackage[normalem]{ulem}
\usepackage{setspace}
\usepackage{mathtools}
\usepackage{gensymb}
\usepackage{amsmath}
\usepackage{epigraph}
\usepackage{shellesc}
\usepackage[style=ieee,
            backend=biber,
            url=false,
            date=year]{biblatex}

\graphicspath{ {./codes/} }

\theoremstyle{remark}

\declaretheorem[name=Objective,sibling=theorem]{objective}

\title{
Kernel Modelling of Fading Memory Systems
}

\author{Yongkang Huo, Thomas Chaffey, Rodolphe Sepulchre% <-this % stops a space
\thanks{The authors are with the University of Cambridge, Department of Engineering, Trumpington Street, CB2 1PZ, \texttt{\{yh415, tlc37\}@cam.ac.uk}.  R. Sepulchre is also with KU Leuven,
Department of Electrical Engineering (STADIUS),
KasteelPark Arenberg, 10,
B-3001 Leuven, Belgium,
\texttt{rodolphe.sepulchre@kuleuven.be}.}
\thanks{The research leading to these results has received funding from the European Research Council under the
Advanced ERC Grant Agreement SpikyControl n.101054323. The work of Y. Huo was supported by the UK Engineering and Physical Sciences Research Council (EPSRC) grant 10671447 for the University of Cambridge Centre for Doctoral Training, the Department of Engineering.  The work of T. Chaffey was supported by Pembroke College, Cambridge.
}% <-this 
}

\begin{document}

\maketitle
\thispagestyle{empty}
\pagestyle{empty}

%%%%%%%%%%%%%%%%%%%%%%%%%%%%%%%%%%%%%%%%%%%%%%%%%%%%%%%%%%%%%%%%%%%%%%%%%%%%%%%%
\begin{abstract}

    The paper is a follow-up of the recently introduced kernel-based framework to  identify nonlinear input-output systems regularised by desirable input-output incremental properties. Assuming that the system has fading memory, we propose to learn the functional that maps the past input to the present output rather than the operator mapping input trajectories to output trajectories. While retaining the benefits of the previously proposed framework, this modification simplifies the selection of the kernel, enforces causality, and enables temporal simulation.

\end{abstract}

%%%%%%%%%%%%%%%%%%%%%%%%%%%%%%%%%%%%%%%%%%%%%%%%%%%%%%%%%%%%%%%%%%%%%%%%%%%%%%%%
\section{Introduction}

Machine learning has boosted our ability to identify nonlinear systems.  Many models have been proposed for nonlinear system identification, both parametric \autocite{boyd_fading_1985,grigoryeva_echo_2018,maass_neural_2000,burghi_system_2020} and non-parametric \autocite{boffi_nonparametric_2021}.  A particular area of recent interest has been in kernel-based identification methods \autocite{Pillonetto2010, Pillonetto2011, Pillonetto2014, Dinuzzo2015, Bouvrie2017, Khosravi2021, Khosravi2023, Ljung2020, Bottegal2016}. It remains a challenge, however, to use identified models for control design, as they often lack guarantees on key system properties such as passivity and small-gain.  This issue has motivated the framework of ``identification for control'' \autocite{Gevers1996}.  Identifying models which satisfy pre-specified system properties has recently received interest \autocite{Revay2021c, Revay2021d}.  System properties such as Lipschitz continuity  have become recognized as important robustness properties in the machine learning literature \autocite{Erichson2023, Winston2020, Revay2020b, Fazlyab2019a}. 

The recent article \autocite{van_waarde_kernel-based_2023} introduced a kernel-based framework for system identification  which allows the data fitting to be regularised with input-output properties specified in the form of incremental integral quadratic constraints (IQCs). The method identifies an input-output operator, but the kernel representation of the operator lacks a state-space realization and simulation forwards in time is expensive as it involves computing an entire future trajectory at each time step. 

This technical note addresses those issues by exploiting the property of {\it fading memory}. The starting  observation is that any time-invariant causal operator is determined by a memory functional that maps the past of the input to the present output. This observation was also the starting point of the seminal paper \cite{boyd_fading_1985} of Boyd and Chua, who defined fading memory as an elementary continuity property for such operators. The standard approach in the literature of fading memory systems has been to represent fading memory operators with state-space models. Here we bypass such state-space representations and consider instead a direct kernel-based representation of the fading memory functional. { In comparison to \cite{van_waarde_kernel-based_2023}, the memory functional learned from input-output data is more practical for simulation and control as its evaluation only involves past inputs and at each time step only computes the current output.} Our main result is to show that this modified framework retains the benefits of the methodology in \autocite{van_waarde_kernel-based_2023}, namely that the learning task can be regularised with desirable incremental input-output properties such as a small Lipschitz constant or incremental passivity.

{ The remainder of this paper is structured as follows. In Section~\ref{sec:preliminary}, we introduce some  notation and background on fading memory.  The novel kernel-based framework is introduced and compared with other frameworks in Section~\ref{chap:kernel}. In Section~\ref{sec:Kernel_example}, we identify simple examples of systems using regularised least squares to demonstrate how incremental small-gain and incremental passivity can be imposed with the proposed framework. Conclusions are drawn in Section~\ref{sec:conclusions}. }

\section{Preliminaries}\label{sec:preliminary} 
 The mathematical setup and notation of this technical note are largely borrowed from \cite{boyd_fading_1985}. The $L_2$ norm in $\R^n$ is denoted by
$|\cdot|$. We consider time-invariant, continuous time, causal systems represented by operators mapping $m$-dimensional input signals to $n$-dimensional output signals. { The input
and output spaces of bounded and
rate-limited (Lipschitz continuous) signals are denoted 
by ${\cal U}\subset C^0_b(\R,\R^m), {\cal Y}\subset C^0_b(\R,\R^n)$, respectively. Input and output signals on the entire time axis $\R$ are denoted by $u_c \in {\cal U}$ and $y_c \in {\cal Y}$, respectively. We denote the space of bounded and
rate-limited past input signals by ${\cal U}_\text{past} \subset C^0_b\left((-\infty,0],\R^m\right)$.
% \begin{IEEEeqnarray*}{lcr}
%  % {{\cal U}}_\text{past}\subseteq\;\{u(\cdot):(-\infty, 0] \to \R^{m}\}.
%  % {\cal U}_\text{past}:={\cal U}_\text{past}\cap {\cal W},\\
%  {\cal U}_\text{past}: \{ u(\cdot):(-\infty, 0] \to \R^{m}  ~| ~|u(t_1)-u(t_2)|<C_1|t_1-t_2|,\\ 
%  \;\;\;\;\;\;\;\;\;\;\;\;\;\;\;\;\;\;\;\;\;\;\;\;\;\;\;\;\;\; ~\sup_t |u(t)|<\infty, 0<C_1<\infty,\},
% \end{IEEEeqnarray*}
}
For $1 \leq p < \infty$, we define $\|u\|_{p}$ for $u:(-\infty, 0]\to\mathbb{R}^m$ by
$$
\|u\|_{p}=\left(\int_{-\infty}^0 |u(t)|^p dt\right)^{1/p},
$$ 
with analogous definitions on the entire real axis for $u_c\in \mathcal{U}$ as
$$
\|u_c\|_{p}=\left(\int_{-\infty}^{\infty} |u_c(t)|^p dt\right)^{1/p}.
$$ 

% where the region of integration is the domain of $u$ $\left(\R \; \text{or} \; (-\infty, 0]\right)$.
The $L_\infty$ norm is given by
\begin{IEEEeqnarray*}{rCl}
        \norm{u}_\infty &=& \sup_t |u(t)|,
\end{IEEEeqnarray*}
where the supremum is taken over the domain of $u$.  

% Unless otherwise stated, we
% assume all signal spaces are equipped with the $L_\infty$ norm.

A \emph{system} is an operator $G: {\cal U} \to {\cal Y}$.   Let $U_\tau$ denote the delay
operator, $(U_\tau u_c)(t):=u_c(t+\tau)$. Then system $G$ is \emph{time
invariant (TI)} if $GU_\tau u_c=U_\tau
Gu_c$. $G$ is said to be \emph{causal} if $u_c(t)=v_c(t)$ for $t\le0$ implies
$(Gu_c)(t)=(Gv_c)(t)$. $G$ is \emph{continuous}
if it is a continuous function $G: {\cal U}\to{\cal Y}$ with respect to the norm of ${\cal U}$ and ${\cal Y}$. 
\textcite{boyd_fading_1985} show that each TI causal operator $G$ has a unique
functional $F:{\cal U}_\text{past}\to \R^n$ mapping past input to current output, such that
$$
Fu\coloneqq(Gu)(0).
$$
We call this functional the \emph{memory functional}.
Each time-invariant and causal operator $G$ is uniquely determined by its memory functional $F$, and vice versa. 

Define $P:{\cal U}\to{\cal U}_\text{past}$ to be the truncation operator, $(Pu_c)(t):=u_c(t) , \forall t\le 0$. Then we have 
$$
(Gu_c)(t)=FPU_{t}u_c.
$$
Since $PU_{t}$ is a continuous operator, $G$ is continuous if and only if $F$ is continuous \cite{boyd_fading_1985}. Hence we say that an operator $G$ is continuous if and only if its memory functional $F:{\cal U}_\text{past}\to \R^{n}$ is continuous: for all $ \epsilon>0$, there exists $\delta>0$ such that, for all $u,v\in{\cal U}_\text{past}$,
$\sup_{t\le0}|u(t)-v(t)|<\delta$ implies $|Fu-Fv|<\epsilon$.

Let $H$ be a map between arbitrary normed spaces $\cal U$ and $\cal Y$.  Then, if, for all $u,
v \in \cal U$, we have
\begin{IEEEeqnarray*}{rCl}
        \norm{Hu - Hv}_{\cal Y} \leq \lambda \norm{u - v}_{\cal U},
\end{IEEEeqnarray*}
$\lambda$ is said to be a \emph{Lipschitz constant} of $H$, and $H$ is said to be Lipschitz. This definition applies equally to operators and  memory functionals.

An operator $H$ is described as having an \emph{incremental small-gain} if its Lipschitz constant is less than or equal to 1 and the two normed spaces are equipped with the $L_2$ norm. This constant serves as a limit on the incremental gain of the system, making it a crucial metric for assessing robustness.

\section{Kernel modelling of fading memory systems}\label{chap:kernel}

\subsection{Fading memory}

Fading memory is an important stability property which expresses that the influence of the past input on the current output fades away over time. An LTI system with the fading memory property has a convolution representation. Fading memory was originally introduced to system approximation by  \textcite{boyd_fading_1985}, where it was shown that Volterra series are universal approximators of systems with fading memory.  The concept later received attention in the context of robust control \autocite{shamma_fading-memory_1993, Sandberg1994}. Recently, \textcite{Gonon2022} introduced the \emph{reservoir kernel} as a universal approximant for the Volterra series expansion of a fading memory operator.

Instead of the original fading memory definition in \cite{boyd_fading_1985}, we adopt and extend the fading 
memory norm  definition of \cite{matthews_identification_1994} and the concept of fading memory normed spaces of \textcite{grigoryeva_echo_2018}. We use this to define the input signal space as a Hilbert space with a fading memory norm, which enables in turn a kernel representation of the memory functional.

{Let $w: (-\infty, 0] \to (0, 1]$ be a weight,
such that $\lim_{t\to-\infty}w(t)=0$ and $w$ has finite $L_p$ norm for all $ 1 \leq p\leq \infty$. We define the $(w, p)$ fading memory norm of a past signal $u$ as 
$$
\|u\|_{w,p}:=\|u\,w\|_{p}, 
$$    
%\end{definition}
%\begin{definition}
 and  the corresponding $(w, p)$ fading memory normed space as 
$$
{\cal W}_{w, p} \coloneqq \{u:(-\infty, 0] \to \R^{m}\; | \;\|u\|_{w, p}< \infty\}.
$$    
%\end{definition}
{
Note that ${\cal U}_\text{past}$ is a subspace of ${\cal W}_{w, p}$.}

In this paper, except Section~\ref{sec:incremental small-gain_pass}, we restrict our attention to the case $p=2$ and denote the norm by $\|u\|_{{\cal U}_\text{past}}:=\|u\|_{w, 2}$.  In this case, we also have an inner product, denoted by $\left\langle u,v \right\rangle_{{\cal U}_\text{past}}=\left\langle u,v \right\rangle_{w}:=\langle uw,vw\rangle$, where $\langle u,v \rangle$ is the regular $L_2$ inner product.} Fading memory is then defined as the continuity with respect to a fading memory norm.
{
\begin{definition}\label{fading memory definition}
Consider an operator $G: {\cal U} \to {\cal Y}$. We say that $G$ has $(w, p)$ \emph{fading memory} if and only if its
functional $F:{\cal U}_\text{past}\to \R^{n}$ is continuous with respect to the $(w,
p)$ fading memory norm: for all $ \epsilon>0$, there exists $\delta>0 $ such that, for
all $u, v \in {\cal U}_\text{past}$, we have
$\|u-v\|_{w, p}<\delta$ implies that $ |Fu-Fv|< \epsilon$.

% \begin{IEEEeqnarray*}{+rCl+x*}
% \|u-v\|_{w, p}<\delta &\Rightarrow&  |Fu-Fv|< \epsilon. 
% &\qedhere
% \end{IEEEeqnarray*}
\end{definition}
}

{\subsection{Problem statement}

{ We address the problem of modelling the fading memory functional $F$ of a fading memory system given a set of input-output data. This is made precise in Objective 1.
\begin{objective}
Given a data set with $l$ data points $\left\{(u_{c,i},y_{c,i})\in {\cal U} \times {\cal Y} \right\}_{i=1}^l$ find $\hat{F}$ that satisfies Definition~\ref{fading memory definition}, such that:
$$\sum_{t=-\infty}^{\infty}\sum_{i=1}^{l}|\hat{F}PU_{t}u_{c,i}-y_{c,i}(t)|^2,$$
is minimized.
\end{objective}
Let $u_{i,t}=PU_{t}u_{c,i}$ and $y_{i,t}=PU_{t}y_{c,i}$. Since in practice, we will not have data with infinite length but instead have data with finite length and finite time samples (e.g. \(t\in\mathcal{T},\text{card}(\mathcal{T})=N\)), we can drop the subscript $t$ by defining:

$$
\{(u_{j},\,y_{j}(0))\} \;:=\; 
\bigl\{\,\bigl(u_{i,t},\,y_{i,t}(0)\bigr): \; i=1,\ldots,l,\; t\in \mathcal{T}\bigr\}.
$$
Then we can rewrite our objective to be:
\begin{objective}
Given $\left\{(u_{j},y_{j}(0))\in {\cal U}_{\text{past}} \times \R^n \right\}_{j=1}^{l\times N}$, find $\hat{F}$ that satisfies Definition~\ref{fading memory definition}, such that:
$$\sum_{j=1}^{l\times N}|\hat{F}u_{j}-y_{j}(0)|^2,$$
is minimized.
\end{objective}

Most importantly, we wish to impose the incremental small-gain property to the approximated system. The extension to other incremental input-output properties is considered in Section \ref{sec:incremental small-gain_pass}.}
}

\subsection{Kernel representation of fading memory functionals}\label{sec:kernel_approx}
 The theory of approximation in Reproducing Kernel Hilbert Spaces is well-established for continuous  functions in real vector spaces
\autocite{paulsen_introduction_2016}, and has been extended to arbitrary separable Hilbert spaces in \autocite{christmann_universal_2010}.  We combine these results with the vector-valued kernels defined in \autocite{micchelli_learning_2005} to approximate any continuous vector-valued functional on a compact subset of a separable Hilbert space. Firstly, we define the RKHS of continuous functionals $\{F:{\cal U}_\text{past}\to\R^n\}$, with the
inner product $\langle\cdot,\cdot\rangle_{{\cal H}}$. 
% We use the notation  $\langle\cdot,\cdot\rangle_{{\cal U}_\text{past}}$ to denote the inner product of past signals in ${\cal U}_\text{past}$ (when it exists) and  $\langle\cdot,\cdot\rangle$ to denote the inner product in $L_2$ or $L_2$.
{
\begin{definition}[RKHS]
    A mapping $k:{\cal U}_\text{past} \times {\cal U}_\text{past} \to {\cal B}(\R^n)$, where ${\cal B}(\R^n)$ denotes the space of bounded linear operators on $\R^n$, is called a \textit{reproducing kernel} of  ${\cal H}$ if :
    \begin{enumerate}
        \item $k(\cdot,u)\alpha:{\cal U}_\text{past}\to\R^n$ is a member of $\cal H$ for all $u\in {\cal U}_\text{past}$ and $\alpha \in \R^n$
        \item The \textit{reproducing property} holds: for all $u\in{\cal U}_\text{past}$, $\alpha \in \R^n$, and $F\in {\cal H}$:
        $$\langle \alpha,Fu\rangle=\langle F,k(\cdot,u)\alpha\rangle_{{\cal H}} 
        $$\qedhere
    \end{enumerate}
\end{definition}
A \emph{universal kernel} is one that can approximate all fading memory functionals on a compact subset $S\in {\cal U}_\text{past}$.

\begin{definition}[Universal kernel \cite{caponnetto_universal_2008}]
Let $C^0({S},\R^{n})$ be the set of all continuous functionals mapping ${S}$ to $\R^n$. Then if ${\cal H}$ is dense in $C^0({S},\R^{n})$, the corresponding kernel $k$ is called a \emph{universal
 kernel} of $C^0({S},\R^{n})$. 
% $\mathcal{U}_{\text{past}}$, there exists a sequence of finite linear combinations of
% the form
% \[
% \sum_{i=1}^m \alpha_i\,K(x_i,\cdot)\quad (m \in \mathbb{N},\;\alpha_i \in \mathbb{R}^n,\;x_i
% \in \mathcal{U}_{\text{past}})
% \]
% that converges uniformly to $F$.  In other words, any continuous $F$ can be approximated
% arbitrarily closely by \emph{finite} sums of kernel sections. This fact follows from
% the Moore–Aronszajn theorem and the definition of universal kernels\cite{paulsen_introduction_2016,christmann_universal_2010,caponnetto_universal_2008}.
\end{definition}

For every continuous functional $F$ on ${S}$, there exists a sequence of finite linear combinations of
the form
\[
\sum_{i=1}^m \alpha_i^\top\,k(x_i,\cdot)\quad (m \in \mathbb{N},\;\alpha_i \in \mathbb{R}^n,\;x_i
\in {S})
\]
that converges uniformly to $F$: for every $\epsilon > 0$, there exists $M$ such that, for all $m \geq M$, there exists $\{(\alpha_i,x_i)\}$, such that for all $y\in S$ $$\left|\sum_{i=1}^m \alpha_i^\top\,k(x_i, y) - F(y)\right| < \epsilon  \qedhere.$$ This fact follows from
the Moore–Aronszajn theorem and the definition of universal kernels \cite{paulsen_introduction_2016,christmann_universal_2010,caponnetto_universal_2008}.}

The following theorem of \textcite{christmann_universal_2010} shows how to find such a kernel for $C^0({S},\R^n)$.
\begin{theorem}[\text{\autocite[Thm. 2.2]{christmann_universal_2010}}]Let $S$ be a compact subset of a separable Hilbert space with inner product $\langle\cdot,\cdot\rangle_{{\cal U}_\text{past}}$. Let $\sigma(\cdot):\R\to\R$ be a function that can be expressed by its Taylor series
$$
\sigma(t)=\sum_{n=1}^{\infty}c_nt^n.
$$
If $c_n>0$ for all $n$ then the kernel $k:{\cal U}_\text{past}\times {\cal U}_\text{past}\to \R$
$$
k(u,v)=\sigma(\langle u,v\rangle_{{\cal U}_\text{past}})
$$
is a universal kernel of $C^0(S,\R)$. 
\end{theorem}
% A universal kernel $k$ in $C^0({\cal U}_\text{past},\R)$ implies that: for any \(\epsilon>0\), there exist a finite set \(S\subset{\cal U}_{past}\) and coefficients \(\{\alpha_{i}\}\) such that
% \[
% \sup_{u\in{\cal U}_{past}}\Biggl|F(u)-\sum_{\hat{u}\in S}\alpha_{i}\, k(\hat{u},u)\Biggr| < \frac{\epsilon}{\sqrt{n}}.
% \]
Any universal kernel in $C^0(S,\R)$ induces a corresponding universal kernel in $C^0(S,\R^n)$.
{
\begin{theorem}
For any universal kernel $k$ in $C^0(S,\R)$ 
$$
k_e(u,v)=k(u,v)I
$$
is a universal kernel in $C^0(S,\R^n)$, where $I$ is the $n\times n$ identity matrix.
\end{theorem}

\begin{proof}
Let \(F\in C^0(S,\R^n)\) with components \(F=(F_1,\dots,F_n)^\top\). For each \(j\in\{1,\dots,n\}\) and any \(\epsilon>0\), universality of \(k\) implies that there exist an integer \(l_j\), centers \(A_j:=\{u_{j,i}\}_{i=1}^{l_j}\subset S\), and coefficients \(\{\alpha_{j,i}\}_{i=1}^{l_j}\subset\R\) such that
\[
\sup_{u\in S} \Biggl|F_j(u)-\sum_{i=1}^{l_j}\alpha_{j,i}\, k(u_{j,i},u)\Biggr| < \frac{\epsilon}{\sqrt{n}}.
\]
Define $\alpha_{j}(u_i)=\alpha_{j,i}, \forall u_i \in A_j$ and let
\[
A:=\bigcup_{j=1}^n A_j,
\]
and enumerate \(\{u_{m}\}_{m=1}^{M}:=A\), where $M=\text{Card}(A)$. For each \(j\), define
\[
\hat{F}_j(u)=\sum_{m=1}^{M}\tilde{\alpha}_{j}(u_m)\, k(u_{m},u),
\]
where \(\tilde{\alpha}_{j}(u_m)=\alpha_{j}(u_m)\) if \(u_{m}\in A_j\) and zero otherwise.
Let \(\hat{F}(u):=(\hat{F}_1(u),\dots,\hat{F}_n(u))^\top=\sum_{m=1}^{M}(\tilde{\alpha}_{1}(u_m),\dots,\tilde{\alpha}_{n}(u_m))\, k(u_{m},u)\). Then
\[
|F(u)-\hat{F}(u)|_2 \le \sqrt{n}\max_{1\le j\le n}\sup_{u\in {S}}|F_j(u)-\hat{F}_j(u)| < \epsilon.
\]
Since \(\hat{F}\) is a finite linear combination of the kernel sections \(k_e(u_m,\cdot)=k(u_m,\cdot)I\), it belongs to the RKHS of \(k_e\). As \(\epsilon>0\) was arbitrary and $M$ is finite, the RKHS of \(k_e\) is dense in \(C^0(S,\R^n)\).
\end{proof}}

A particular example of a universal kernel of {$C^0(S,\R^n)$} is the Gaussian-type
Radial Basis Function (RBF) kernel with $\sigma\neq0$:
$$
k_e(u,v)=I\exp\left(-\frac{\|u-v\|_{{\cal U}_\text{past}}^2}{2\sigma^2}\right).
$$
Using these results, if the input set $S$ is a compact subset of ${\cal U}_\text{past}$, we can use kernel methods to approximate any functional in $C^0(S,\R^{n})$ with a universal kernel. In the Appendix, we show that the set of bounded and rate-limited signals with a finite upper bound $C_2$ and Lipschitz constant $C_1$:
{
$$S_b=\{u\in {\cal U}_\text{past}|~|u(t_1)-u(t_2)|<C_1|t_1-t_2|,\sup_t |u(t)|<C_2\}$$
}
for some $0<C_1 , C_2<\infty$ and all $t_1,t_2\leq 0$, is a compact subset of ${\cal U}_\text{past}$ with a fading memory norm. The same set of signals was adopted in the original work of Boyd and Chua. 

\subsection{Kernel regularised learning of memory functional}

Kernel methods are popular in machine learning for their ability to {\it regularize} the data fit of the approximation.  

Given a finite data set of past inputs and current outputs, together with a fading memory kernel, the representer theorem provides an optimal approximant with the fading memory property.

\begin{theorem}[Representer Theorem \protect{\autocite[Thm. 8.7]{paulsen_introduction_2016}}]
Given  $l$ past inputs and current outputs $(u_i,y_i(0)) \in S_b \times \R^n$, let $\hil$ be an RKHS of functionals from $S_b$ to  $\R^n$. Then   the regularised least-squares problem
$$
\min_{\hat F \in {\cal H}} \sum_{i=1}^{l}|\hat{F}u_{i}-y_i(0)|^2+\gamma\|\hat{F}\|^2_{\hil}.
$$ 
has a unique solution given by
$$
\hat{F}u=\sum_{i=1}^{l}\alpha_i^\top k(u_{i},u),\quad \hat{\alpha}=(K+\gamma I)^{-1}Y,
$$
{ where $\hat{\alpha}$ is a matrix whose rows are $\alpha_i$, $K_{i,j}=k(u_{i},u_{j})$ and $Y$ is a matrix whose rows are $y_i(0)$.}
\end{theorem}

The regularised least squares problem  minimizes the sum of an empirical loss and the norm of the functional.  Minimizing the empirical loss provides a good data fit whereas minimizing the functional norm gives better generality and robustness to noise in the data. The regularization parameter $\gamma$ controls the trade-off between data fit and regularization. 

With regularization on the fading memory functional, we can also impose the incremental small-gain property on the learned system. In this part, we will first show the relationship between the Lipschitz constant of the functional and the incremental gain of the system and hence prove that by regularizing the Lipschitz constant of the approximated functional, we can ensure that the approximated system has incremental small-gain.

First, we define \emph{Lipschitz Kernels}.
\begin{definition}A kernel $k(u,v)$ is said to have \emph{Lipschitz constant $r$} if
 \begin{equation}
 \|k(u,u)-2k(u,v)+k(v,v)\|_{{\cal B}(\R^n)}\le r^2\|u-v\|_{{\cal U}_\text{past}}^2, \label{eq:continous_kernel}
 \end{equation} for some $r>0$.
 It is said to be a nonexpansive kernel in \cite{van_waarde_kernel-based_2023} when $0< r\le1$.
 \end{definition}
In addition, if a kernel has Lipschitz constant $r$, then
 $
 |\hat{F}u-\hat{F}v|\le r\|\hat{F}\|_{\hil }\|u-v\|_{{\cal U}_\text{past}}.
 $  
This is a direct application of \autocite{van_waarde_kernel-based_2023}[Lemma 1].
 Van Waarde and Sepulchre \autocite{van_waarde_kernel-based_2023}[Prop. 2] provide several examples of kernels which satisfy property~\eqref{eq:continous_kernel}. In particular, the Gaussian-type RBF-kernel is a universal kernel with Lipschitz constant $r=\frac{2}{\sigma^2}$. We now give a condition on the Lipschitz constant of the functional which guarantees an incremental small-gain operator.

 {
\begin{theorem} [Incremental small-gain regularization]\label{cor:regularization}
Let ${\cal U}_\text{past}$ be equipped with the norm $\left\|\cdot\right\|_{{\cal U}_\text{past}}^2=\left\|\cdot\right\|_{w,2}^2$. Let $c \coloneqq \left(\int_{-\infty}^{0}w(t)^2dt\right)^{\frac{1}{2}}$ and $G: {\cal U} \to {\cal Y}$ with a memory functional $F:{\cal U}_\text{past}\to \R^n$ that is Lipschitz continuous with respect to the norm $\left\|\cdot\right\|_{{\cal U}_\text{past}}$ with Lipschitz constant $\beta$, then $G$ has incremental small-gain if $\beta^2<\frac{1}{c^2}$.
\end{theorem}

\begin{proof}
Let $u_c,v_c\in{\cal U}$ and $ \left\|u_c\right\|_2,\left\|v_c\right\|_2<\infty$, let $u(t,\cdot)$ and $v(t,\cdot)$ denote the truncated $u_c$ and $v_c$ at time $t$ such that 
$$
u(t,\tau)=
\begin{cases}
    u_c(t+\tau) & \tau\le 0 \\
    0 & \tau>0.
\end{cases}
$$
It follows that $Pu_{t}, Pv_{t}\in {\cal U}_\text{past}$.  We then have
\begin{IEEEeqnarray*}{rCl}
\|Gu_c-Gv_c\|_2^2&=&\int_{-\infty}^{\infty}|Gu_c(t)-Gv_c(t)|^2 dt\\
&=&\int_{-\infty}^{\infty}|FPu(t,\cdot)-FPv(t,\cdot)|^2 dt\\
&\leq&\int_{-\infty}^{\infty}\beta^2\|Pu(t,\cdot)-Pv(t,\cdot)\|_{{\cal U}_\text{past}}^2 dt\\
\end{IEEEeqnarray*}
% &=&\int_{-\infty}^{\infty}\beta^2 \|u_{t}-v_{t}\|_{w,2}^2 dt
\vspace{-20pt}
%\begin{eqnarray}
$$
=\int_{-\infty}^{\infty}\beta^2 \int_{-\infty}^{0}|u_{c}(t+\tau)-v_{c}(t+\tau)|^2w(\tau)^2d\tau dt.
$$
Let $w'(\tau)=w(\tau)$ for $\tau\le0$ and $w'(\tau)=0$ for $\tau>0$. Then
$$\|Gu_c-Gv_c\|_2^2\le\int_{-\infty}^{\infty}\beta^2 \int_{-\infty}^{\infty}|u_c(\tau)-v_c(\tau)|^2w'(t+\tau)^2d\tau dt$$
Since $w'$ is a square integrable function, by Fubini's theorem,
$$\|Gu_c-Gv_c\|_2^2\le\int_{-\infty}^{\infty}\beta^2 |u_c(\tau)-v_c(\tau)|^2\int_{-\infty}^{\infty}w'(t+\tau)^2dt d\tau.$$
Note that $\int_{-\infty}^{\infty}w'(t+\tau)^2dt=c^2$, so we have
$$\|Gu_c-Gv_c\|_2^2\le c^2 \int_{-\infty}^{\infty}\beta^2 |u_c(\tau)-v_c(\tau)|^2d\tau
=\beta^2c^2\|u_c-v_c\|_2^2.$$
If $\beta^2<\frac{1}{c^2} $, it follows that the operator $G$ has incremental small-gain.
\end{proof}
}
To impose incremental small-gain on the approximated system, as described in Theorem \ref{cor:regularization}, it suffices to choose a sufficiently large value for $\gamma$. Obviously, there exists a trade-off between the size of the bias $\gamma$  and the incremental  gain of the approximated system: enforcing a smaller gain requires  a larger bias. We can also use this constraint to impose other IQC properties via the scattering transform, for more details, we refer the reader to \cite{van_waarde_kernel-based_2023}.
\begin{figure}[H]
    \centering
    \includegraphics[scale=0.4]{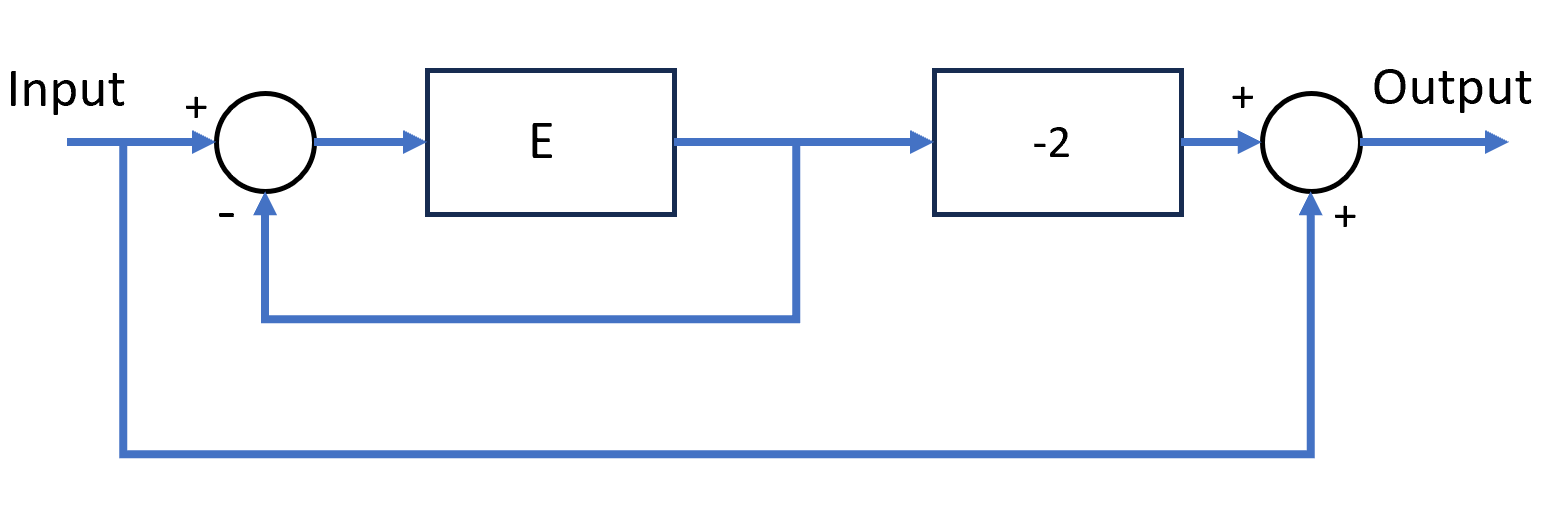}
    \caption{Block diagram implementation of scattering transform in \cite[Fig.~4(a)]{anderson_small-gain_1972}. $E$ is a incremental small-gain system and the overall system is an incremental passive system}\label{fig:scattering}
\end{figure}
\subsection{Regularised learning with IQC constraints} \label{sec:incremental small-gain_pass}

A system $G$ is \emph{incrementally passive} if its input space ${\cal X}$ and output space ${\cal Y}$ are the same Hilbert space and $\langle Gx_1-Gx_2, x_1-x_2\rangle \geq 0$ for all $x_1,x_2\in {\cal X}$
\begin{theorem}[Scattering transform between incremental small-gain system and incrementally passive system]
    If a system $E$ is an incremental small-gain and its input space ${\cal X}$ and output space ${\cal Y}$ are the same Hilbert space, then $G:=(I-E)(I+E)^{-1}$ is incrementally passive. If a system $G$ is incrementally passive, then $E:=(G-I)(G+I)^{-1}$ has incremental small-gain.
\end{theorem}
Similarly to \cite{van_waarde_kernel-based_2023}, we can convert the small-gain constraint to other IQCs via a scattering transform. For more detail, we refer the interested reader to \cite{van_waarde_kernel-based_2023}[Thm. 4]. Due to space limitations, in this paper, we will focus on the illustration of incremental passivity. The implementation of the scattering transform can be accomplished through a feedback circuit \cite{anderson_small-gain_1972}, as depicted below:
To impose incremental passivity, we can utilize the scattering transform to convert the input-output data of system $G$ into the input-output data of the incremental small-gain system $E$. 
Let $\{u_i,y_i\}^N_i$ represent $N$ data points from system $G$, where $u_i,y_i \in \mathcal{U}_\text{past}$. After applying the scattering transform to the data points, we obtain new data point pairs $\{u_i',y_i'\}^N_i$, where $u_i'=\frac{1}{2}(u_i+y_i)$ and $y_i'=\frac{1}{2}(u_i-y_i)$. To acquire data points for learning the fading memory functional, which maps the past signal to the current output, we truncate the $y_i$ to $y_i(0)$. Next, we fit this transformed data to the kernel model using a sufficiently large value of $\gamma$ to ensure incremental small-gain. Finally, by connecting $E$ as depicted in Figure \ref{fig:scattering}, we can recover the original system $G$, benefiting from the incremental small-gain property of $E$, resulting in an incrementally passive approximation. However, for a strictly proper LTI system $G$, the transformed system $E$ might not be strictly proper. Hence this system is not continuous with respect to the $L_2$ fading memory norm defined previously. {Instead of equipping the input space with the $L_2$ fading memory
norm, we define ${\cal W_\text{union}}\subset{\cal W}_{w,2}\cap{\cal W}_{\delta,\infty}$ with norm $ a\|\cdot\|_{w,2}^2 + b\|\cdot\|_{\delta,\infty}^2$, where $\delta(0)=1$, $\delta(t)=0$ for $t<0$ and equip ${\cal U}_\text{past}$ with this norm. Compared to the normal $L_2$ norm, this norm takes into account the instantaneous effect at time 0, which is the property of a biproper system. This norm can be expressed as an $L_2$-type norm combining an $L_2$ term and a pointwise term as $\|x\|^2=a\|x\|_{w,2}^2 + b|x(0)|^2$. With slight modifications, Theorem \ref{cor:regularization} remains applicable to this norm as well.
\begin{corollary}
    Let the input space ${\cal U}_\text{past}$ be equipped with norm $\|x\|^2=a\|x\|_{w,2}^2 + b|x(0)|^2$, $c \coloneqq \left(\int_{-\infty}^{0}w(t)^2dt\right)^{\frac{1}{2}}$. Suppose a system $G$ has a memory functional $F:{\cal U}_\text{past}\to \R^n$ that is Lipschitz continuous with respect to the above norm with Lipschitz constant $\beta$. Then
    $$
    \|Gu_c-Gv_c\|_{L_2}^2<\beta^2(b+ac^2)\|u_c-v_c\|_{L_2}^2,$$ and hence $G$ has incremental small-gain if $\beta^2<\frac{1}{b+ac^2}$.
\end{corollary}The proof of the above corollary is the same as Theorem \ref{cor:regularization}. Also note that the set $S_b$ is compact with respect to this norm as it is compact with respect to $\|\cdot\|_{w,2}$ (shown in the Appendix) and $\|\cdot\|_{w,\infty}$ (shown in \cite{boyd_fading_1985}), hence we can use a kernel associated with this norm to approximate the memory functional.} The bilinear kernel and the RBF kernel associated with this norm are defined as follows, for positive parameters $a$ and $b$.
 \begin{definition}{Modified bilinear kernel}
    $$k(x,y)=a\langle x,y\rangle_{w,2} + bx(0)y(0) \qedhere$$
\end{definition}
\begin{definition}{Modified RBF kernel}
    $$k(x,y)=\exp\left(\frac{-(a\|x-y\|_{w,2}^2 + b|x(0)-y(0)|^2)}{2\sigma^2}\right). \qedhere$$
\end{definition}

\subsection{Comparison with previous frameworks}

{ Most previous models \cite{boyd_fading_1985,maass_neural_2000,maass_computational_2004,grigoryeva_echo_2018} that approximate fading memory systems rely on the same conceptual construction: compactness of the set of past signals $S_b$  guarantees uniform convergence of the approximation by the classical Stone-Weierstrass theorem. The methods of approximation then follow these steps:
\begin{enumerate}
    \item Project data points from compact set $S_b$ to $\mathbb{R}^N$, where $N$ is chosen large enough to define a bijective continuous map $\phi: S_b \to \mathbb{R}^N$.
    \item Approximate the function from the projected data in $\mathbb{R}^N$ to the outputs in $\mathbb{R}^{n}$ using a universal approximator, such as a polynomial, sigmoid gate, or kernel. This produces a function $\hat{f}: \mathbb{R}^N \to \mathbb{R}^N$.
    \item Obtain the approximated functional $\hat{F} = \hat{f} \circ \phi$.
\end{enumerate}

All the methods above suffer from the same limitation: the bijection between the infinite dimensional set of signals and the finite-dimensional state-space is defined theoretically but not algorithmically. Secondly, all the methods suffer the limitation of nonlinear state-space models to capture incremental input-output properties. This limitation has been a key motivation for the present work, see \cite{sepulchre_incremental_2022} for further details.

A key novelty  in the present paper is to bypass the problem of
finding a bijection from the data set to state-space $\phi: S_b \to \R^N$. Furthermore, as previously discussed and will be continually discussed in the rest of the paper, the kernel framework provides a versatile approximation methodology to encode  input-output properties in the model structure.
}

\section{Examples}\label{sec:Kernel_example}
In this section, we illustrate the proposed method on elementary examples\footnote{The code for these examples can be found at \url{https://github.com/Huoyongtony/regularised-Learning-of-Nonlinear-
Fading-Memory-Systems}.}.
\subsection{LTI examples}
\begin{example}[Choosing different fading memory normed spaces and regularization constants]
To illustrate the role of the fading memory norm, we apply the kernel approximation to the LTI system $$H(s)=\frac{s+1}{(s+3)(s+10)}.$$ Using our prior knowledge that the system is linear, we choose a linear model defined by a bilinear kernel. We assume the input space is a fading memory normed space with an exponential fading memory weight, $w(t)=e^{\frac{\lambda}{2} t}$: $$\|u\|_{{\cal U}_\text{past}}=\|u\|_w=\left(\int_{-\infty}^{0}u(t)^2e^{\lambda t}dt\right)^{\frac{1}{2}}. $$ The corresponding inner product is $$\langle u,v \rangle_w=\int_{-\infty}^{0}u(t)v(t)e^{\lambda t}dt,$$ which also defines the {\it canonical} bilinear kernel. The following experiment shows the effect of varying the rate $\lambda$. 

To generate the training data set, we simulate the LTI system with 100 sine waves of different angular frequencies ($e^{-5}$ to $e^5$ rad/s). We consider the output value of each response at $t = 2$ s and the corresponding past input in the last two seconds as training data. We then test the step response of the approximated system. The results of this experiment are shown in Figure~\ref{fig:full}.  If the fading memory decays more slowly than the slowest pole of the system, given enough data, the approximation is always accurate. If the fading memory  decays faster than the slowest pole, the approximation becomes less accurate but still captures the fast dynamics. If the fading memory decays faster than all poles, the approximation error is large. We then repeat the experiment with fixed $\lambda=1$, and vary the regularization constants. The result is illustrated in Figure~\ref{fig:full_reg}. {As expected, stronger regularization enforces a smaller $L_2$ gain in the approximated system.}

\begin{figure}
    \centering
    \includegraphics[scale=0.6]{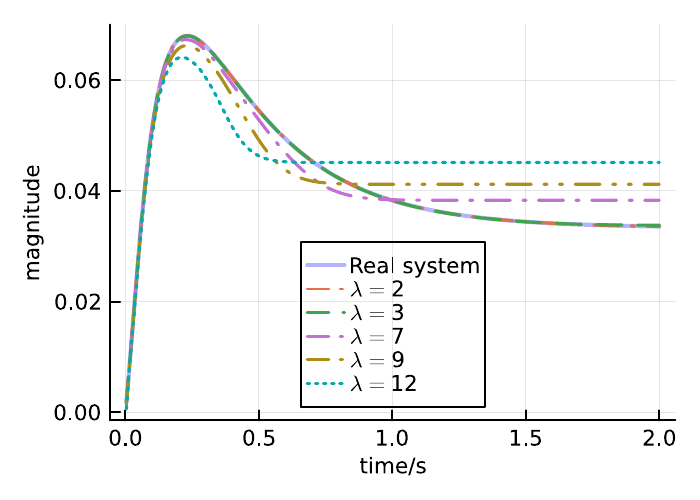}
    \caption{Example 1. Approximating LTI with linear kernel under different normed spaces, solid line is real system, dashed lines are approximated systems.}\label{fig:full}
\end{figure}
\begin{figure}
    \centering
    \includegraphics[scale=0.6]{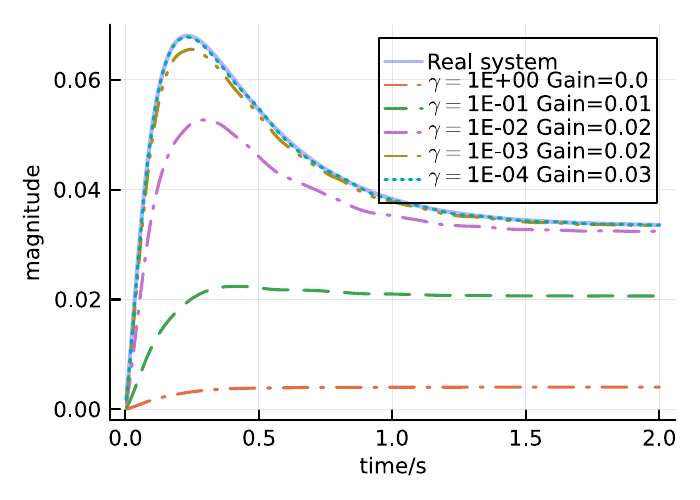}
    \caption{Example 1. Approximating LTI with linear kernel under different regularization constants, solid line is real system, dashed lines are approximated systems.}\label{fig:full_reg}
\end{figure}

\end{example}
\begin{example}[Enforcing passivity on a non-passive linear system]
{In this example, we apply our framework to identify a passive linear system from data generated by the non-passive system 
$$H(s)=\frac{10}{(s+1)(s+2)}.$$ For this we need to use the modified bilinear kernel mentioned in~\ref{sec:incremental small-gain_pass} due to the scattering transform.

To generate the training data set, setting all initial states to zero, we simulate system $H$ with one impulse and 100 sine waves of different frequencies ($e^{-1}$ to $e^3$ rad/s) with unit amplitude for 30 seconds. We consider the output value of each response at $t = 30$ s and the corresponding past input in the last 4 seconds as training data. Afterward we transform the data with the scattering transform and learn the new system with a small-gain constraint. We set $a=b=1$, $\gamma=14.8$ and $\lambda=0.5$. In this setting the system norm is $0.999$. We then recover the system with the scattering transform  as shown in Figure~\ref{fig:scattering}, and plot the Nyquist plot of the approximated system, as illustrated in Figure~\ref{fig:scatter_linear}.} We see that the approximated system is not just passive but about 1 unit from the imaginary axis. This means the bound we derived is a bit conservative.

\end{example}

\begin{figure}
    \centering
    \includegraphics[scale=0.6]{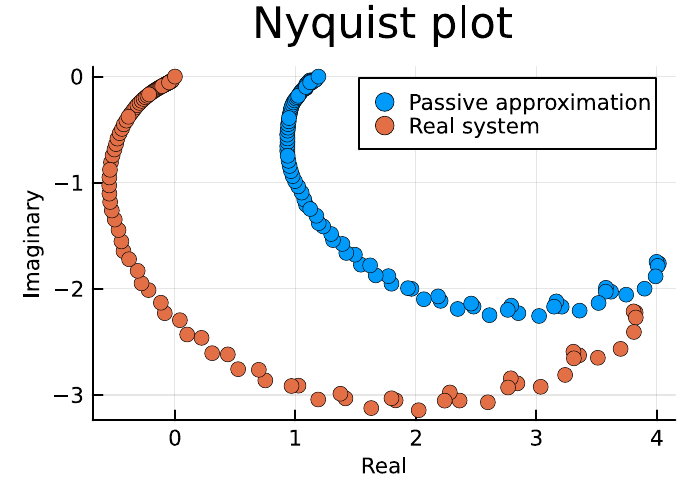}
    \caption{Example 2. Approximating LTI with linear kernel and incremental passivity constraint.}\label{fig:scatter_linear}
\end{figure}

\subsection{Nonlinear examples}
\begin{example}[Kernel regularised learning of nonlinear systems]\label{simple nonlinear system}
We consider a saturated first-order lag with equation
\begin{IEEEeqnarray*}{rCl}
\dot{x} &=& -5x+u,\\
y&=&\operatorname{sat}(x),\\
\operatorname{sat}(x)&=&
\begin{cases}
-5 & x< -5\\
x & |x|\le5\\
5 & x> 5.
\end{cases}\\
% \epsilon &\sim&N(\mu=0,\sigma^2=0.1)
\end{IEEEeqnarray*}
% We also add multiplicative white noise with variance $0.1$ to the output. 
We choose an exponential fading memory with $\lambda=4$.
Again, we probe the system with sine waves of different frequencies (angular frequency 0-30 rad/s) and amplitudes (0-50). We take three points from the output, at times $0.666$ s, $1.333$ s and $2$ s, and their corresponding past input in the last 2 seconds from each response as training data. The system was approximated with a Gaussian RBF kernel ($\sigma=0.25$). 
\end{example}

\begin{figure}
    \centering
    \includegraphics[scale=0.6]{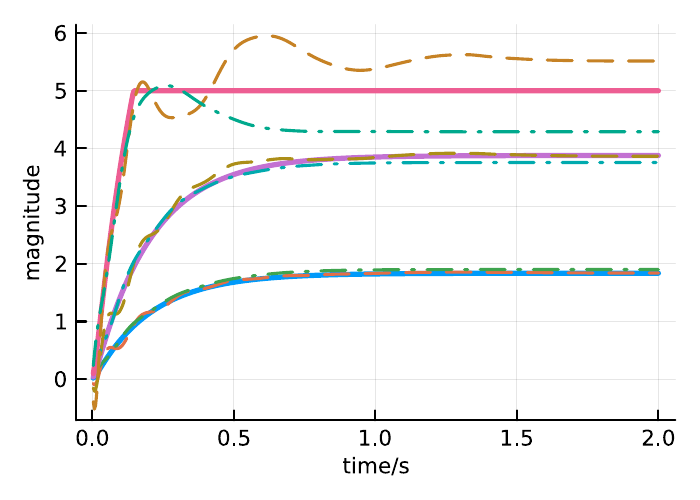}
    \caption{Example 3. Solid lines are outputs of the original system. Dashed lines are outputs of an RBF kernel approximation under $\gamma=0.01$. Dash-dot lines are outputs of an RBF kernel approximation under $\gamma=1e^{-5}$.}
    \label{fig:reg1}
\end{figure}

Figure~\ref{fig:reg1} shows the output of the original and approximate systems when subject to step inputs, for two different values of the regularization constant $\gamma$.  
The step response experiments illustrate that a higher regularization constant favors a smoother input-output behaviour. This is expected as we have proved that regularization can control the Lipschitz constant $\beta$. Moreover, a higher regularization constant gives better generality and better robustness to noise in the data. For the value $\gamma=0.01$, the Lipschitz constant of the approximated operator is about 11.

\begin{example}[Enforcing incremental passivity on a non-passive nonlinear system]
{The system in Example \ref{simple nonlinear system} is not incrementally passive \cite{Kulkarni2001}. We construct an incrementally passive approximation as follows. 

Setting all initial states to zero, we collect the responses of the system to 10 sine wave inputs, each with frequency 10 Hz and duration 2 s. Each sine wave has different magnitude, from 15 to 150. The simulation time step is 0.002 s. We then sample these responses every 10 time steps and apply the scattering transform to generate the training data. Afterward, we learn an incremental small-gain system with a modified RBF kernel mentioned in~\ref{sec:incremental small-gain_pass} on the scattering transformed data.}

{In this case, we set $\sigma=1$, $a=b=0.1$, $\gamma=0.0005625$ and $\lambda=10$. We scale both our input and output data after the scattering transform by $0.0005$. This  scaling is used by \cite{van_waarde_kernel-based_2023} so that the algorithm gives a less conservative bound. 
%Currently, we have no theory on how to choose this scale to achieve a less conservative bound but this might work out in future researches. 
In this setting, the trained model has $\beta^2(b+ac^2)=0.998$ implying that the model has incremental small-gain.

Finally, we use one of the training data to test if we can recover the incrementally passive approximation using the scattering transform, by putting the trained operator in the feedback loop shown in Figure~\ref{fig:scattering}.
The result is shown in Figure~\ref{fig:recovered_operator}.} We can see that the output of the incrementally passive approximation is similar to the output of the linear part in the non-incrementally passive nonlinear system with a small phase shift.

\begin{figure}
    \centering
    \includegraphics[scale=0.6]{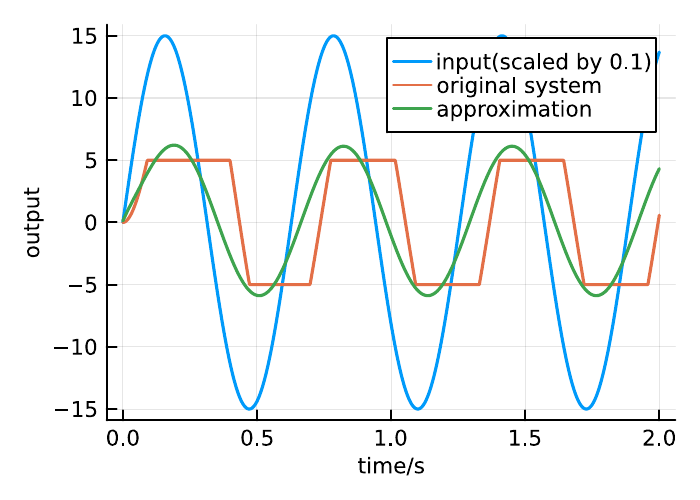}
    \caption{Example 4. Approximating Nonlinear system with RBF kernel and incremental passivity constraint. Response to sine wave with frequency 10 Hz and amplitude 150}\label{fig:recovered_operator}
\end{figure}

\end{example}
\section{Conclusion}\label{sec:conclusions}
{\color{black}
In this paper, we proposed a kernel-based framework for modelling and identification of nonlinear fading memory systems.  The proposed framework departs from the state-space representations that have dominated the literature on fading memory systems. In contrast, we use a kernel to model  the memory functional. Kernel approximations are well suited for regularised learning and to encode prior knowledge about input-output properties. We showed how regularization can be used to impose IQC properties. Future work will explore how to enforce arbitrary incremental integral quadratic constraints by using regularization techniques and/or special kernels.
}
\appendices
{\color{black}
\section{Compactness with respect to the fading memory norm}
\textcite{boyd_fading_1985} observed that the set of bounded and rate-limited input signals with a finite upper bound is compact with respect to a fading memory norm.  In this section, we develop this result for arbitrary fading memory normed spaces.

The following lemma generalizes \autocite[Corollary 2.7]{grigoryeva_echo_2018} to $L_2$ norms in continuous time.
{
\begin{lemma} $S_b$ is a compact subset of the normed space with the norm $\|\cdot\|_{w,2}$. \end{lemma}
\begin{proof} 
To show this, we need to show that every sequence in $S_b$ has a convergent sub-sequence in the sense of $\|\cdot\|_{w,2}$.  \textcite[Lem. 1]{boyd_fading_1985} shows that every sequence in $S_b$ has a sub-sequence which converges in the infinity norm $\|\cdot\|_{w,\infty}$ when the output dimension $m=1$. This can be extended to $m>1$ by combining the converging sub-sequence in each coordinate to construct a sub-sequence that converge in all coordinates. We now show that such a sub-sequence converges in the $L_2$ norm as well.
Consider a sequence $u_{n} \in S_b$, and let $u_{n_k}$ be a converging sub-sequence: $\exists u_{0} \in S_b$  
$$
\|u_{n_k}-u_{0}\|_{w,\infty}\to 0, k\to \infty.
$$
By Hölder's inequality,
$$
\|1(u_{n_k}-u_{0})^2\|_{w,1}\leq\|1\|_{w,q/(q-2)}\|(u_{n_k}-u_{0})^2\|_{w,q/2},
$$
for any $2<q<\infty$. Since we defined $w$ to have finite $L_q$ norm for all $1\leq q \leq \infty$, we can let $C= \sup_q{\|1\|_{w,q/(q-2)}}$. Rearranging, we get,
$$\|u_{n_k}-u_{0}\|_{w,2}\leq C^{\frac{1}{2}}\|u_{n_k}-u_{0}\|_{w,q},$$
$$\|u_{n_k}-u_{0}\|_{w,2}\leq  \lim_{q\to \infty} C^{\frac{1}{2}}\|u_{n_k}-u_{0}\|_{w,q}=C^{\frac{1}{2}}\|u_{n_k}-u_{0}\|_{w,\infty}.$$
Hence, $S_b$ is a compact subset with respect to $\|\cdot\|_{w,2}$.\qedhere
\end{proof}
}

\printbibliography
}
\end{document}